\documentclass[12pt]{article}
\usepackage{amsmath}
\usepackage{amssymb}
\usepackage{amsthm}
\usepackage{enumerate}
\usepackage{amsthm}
\numberwithin{equation}{section}
\newtheorem{theorem}{Theorem}[section]
\newtheorem{definition}[theorem]{Definition}
\newtheorem{proposition}[theorem]{Proposition}
\newtheorem{corollary}[theorem]{Corollary}
\newtheorem{lemma}[theorem]{Lemma}

\newtheorem{remark}[theorem]{{\bf Remark}}

\def\cH{{\mathcal H}}

\newcommand{\Norm}{\|\,\raisebox{0.5 ex}{.}\,\|}

\newcommand{\Skindef}{[\raisebox{0.5 ex}{.},\raisebox{0.5 ex}{.}]}
\newcommand{\Skdef}{(\raisebox{0.5 ex}{.},\raisebox{0.5 ex}{.})}
\newcommand{\PT}{{\cal PT\,}}
\newcommand{\TP}{{\cal TP\,}}
\newcommand{\Pt}{{\cal P\,}}

\def\dR{{\mathbb R}}

   \def\cD{{\mathcal D}}   \def\cH{{\mathcal H}}

\def\cP{{\mathcal P}}    \def\cT{{\mathcal T}}

\def\dom{{\rm dom\,}}

\def\mmm{{{\rm max}}}

  \def\dom{{\rm dom\,}}
  \def\dim{{\rm dim\,}}
  \def\ker{{\rm ker\,}}

\begin{document}

\title{$\PT$ Symmetric, Hermitian and $\mathcal P$-Self-Adjoint Operators Related to
Potentials in $\PT$ Quantum Mechanics}

\author{Tomas Ya.\ Azizov
and Carsten Trunk}

 \maketitle

\begin{abstract}
In the recent years a generalization
 $H=p^2 +x^2(ix)^\epsilon$ of the harmonic oscillator using a complex deformation was investigated, where $\epsilon$ is a real
parameter.
Here, we will consider the most simple case:  $\epsilon$ even and $x$ real.
We will give a complete characterization
of three different classes of operators associated with the differential  expression $H$:
The class of all  self-adjoint (Hermitian) operators, the class of all
 $\PT$ symmetric operators  and the class of all $\mathcal P$-self-adjoint
 operators. Surprisingly, some of the $\PT$ symmetric operators associated
 to this expression have no resolvent set.
\end{abstract}

\section{Introduction}

In the well-known paper  \cite{BB98} from 1998 C.M.\ Bender and S.\
Boettcher considered the following Hamiltonians $\tau_{\epsilon}$,
\begin{equation}\label{Foundation}
\tau_{\epsilon} (y) (x) := - y^{\prime \prime}(x) + x^2
(ix)^{\epsilon} y(x), \quad {\epsilon} >0.
\end{equation}
This gave
rise to a mathematically consistent complex extension of
conventional quantum mechanics  into $\PT$ quantum mechanics, see,
 e.g., the review paper \cite{Bender}. During
the past ten years these $\PT$ models have been analyzed intensively.

Starting from the pioneering work of C.M.\ Bender and S.\ Boettcher
\cite{BB98}, the above Hamiltonian $\tau_{\epsilon}$ was always
understood as a complex extension of the harmonic oscillator
$H=\frac{d^2}{dx^2} + x^2$ defined along an appropriate complex
contour within  Stokes wedges. In \cite{M05} the problem was mapped
back to the real axis using a real parametrization of a suitable
contour within the Stokes wedges and in \cite{alle,JM1,JM2} this
approach was extended to different parametrizations and contours.

Usually, see, e.g., \cite{BB98,Bender,BBM99,CCG,CGS05}, a closed densely
defined  operator $H$ in the Hilbert space $L^2(\dR)$ is called
$\PT$ symmetric if $H$ commutes with $\PT$, where
$\cP$ represents parity reflection and the operator $\cT$ represents
time reversal, i.e.
\begin{equation}\label{PTDef}
(\cP f)(x) = f(-x)\quad \mbox{and} \quad (\cT f)(x) =
\overline{f(x)}, \quad f\in L^2(\dR).
\end{equation}
Via the parity operator $\cP$ an indefinite inner product is given by
 \begin{equation*}
[f,g]:=\int_\dR f(x)\overline{(\cP g)(x)}\,dx = \int_\dR
f(x)\overline{g(-x)}\,dx, \quad f,g\in L^2(\dR).
\end{equation*}
With respect to this inner product, $L^2(\dR)$ becomes a Krein
space and, as usual, a closed densely
defined  operator $H$  is called
$\cP$-self-adjoint if $H$ coincides with its $\Skindef$-adjoint,
see, e.g., \cite{AI,B,Krein,L65,LT04}.
For the study of $\PT$ symmetric operators in the frame work of 
self-adjoint operators in Krein spaces we refer to 
 \cite{LT04,M1,AK1,J02,GSZ,T06,T06a}.

 For unbounded
operators both notions, $\PT$ symmetry and $\cP$-self-adjointness, are also conditions on the domains. These two
notions will be of central interest in this paper, therefore we emphasize them
 in the following definition.
 We denote by $\dom H$ the
domain of the operator $H$.

\begin{definition} \label{PTDEF}
A closed densely defined  operator $H$ in $L^2(\dR)$
 is said to be $\PT$ symmetric if for all $f \in \dom H$ we have
$$
\PT f \in \dom H \quad \mbox{and} \quad \PT Hf = H\PT f.
$$
It is called  $\cP$-self-adjoint if  we have
$$
 \dom H = \dom H^*\cP  \quad \mbox{and} \quad Hf=\cP H^* \cP f \quad
 \mbox{for } f \in \dom H.
$$
\end{definition}
Clearly, a $\cP$-self-adjoint operator $H$ is also $\cP$-symmetric, that is, we have
$$
[Hf,g] =[f,Hg] \quad
 \mbox{for all } f,g \in \dom H.
$$

Here we will restrict ourselves to the most simple case: We will
consider the differential
expression $\tau_\epsilon$ in  \eqref{Foundation}  only for real $x$.
Moreover, if $\epsilon$ is even, we obtain a real-valued potential, i.e.\ if
$\epsilon =4n$ the above differential expression $\tau_{4n}$,
$n\in \mathbb N$, in \eqref{Foundation} will be of the form
\begin{equation}\label{EBurg1}
\tau_{4n} (y) (x) := - y^{\prime \prime}(x) + x^{4n+2} y(x),  \quad x\in \mathbb R.
\end{equation}
and  it will be of the form
\begin{equation}\label{EBurg2}
\tau_{4n+2} (y) (x) := - y^{\prime \prime}(x) - x^{4n+4}
 y(x), \quad x\in \mathbb R.
\end{equation}
in case $\epsilon=4n+2$.
In this situation we can make use
of the well-developed theory of Sturm-Liouville operators
(see, e.g., \cite{LS,W,W03,Z05}). Namely, it turns out that
the expression $\tau_{4n}$ is in the limit point case at $\infty$ and
at $-\infty$, hence there is only one self adjoint operator connected to
$\tau_{4n}$ which is also $\PT$ symmetric and $\cP$-self-adjoint.

The more interesting case is $\epsilon=4n+2$. The differential
expression  $\tau_{4n+2}$  is then in limit circle case at $+\infty$ and
at $-\infty$ and it admits many different extensions.
These extensions are described via restrictions of the
maximal domain  $\cD_{\mmm}$ by ``boundary conditions at $+\infty$
and $-\infty$''.
Therefore, we will
consider the differential expression $\tau_\epsilon$ only in the
case of $\epsilon =4n+2$, $n\in \mathbb N$.

Actually, we will consider
a slightly more general case which includes the case of $\tau_{4n+2}$.
For this,
we will always assume that $q$ is a real valued function from
 $L_{loc}^1(\mathbb R)$ which is even, that is,
 $$
 q(x)=q(-x)\quad  \mbox{for all }
 x\in \mathbb R,
 $$
  such that the differential equation
\begin{equation}\label{FoundationNew}
\tau_{q} (y) (x) := - y^{\prime \prime}(x) - q(x)
 y(x), \quad x\in \mathbb R.
\end{equation}
is in limit  circle case at $+\infty$ and
$-\infty$.

It is the aim of this paper
to specify three classes of operators
connected with the
differential expression $\tau_{q}$ in (\ref{FoundationNew}):
$\PT$ symmetric operators, $\mathcal P$-self-adjoint operators
and self-adjoint (Hermitian) operators. The main result of
this paper is a full characterization of these classes, which,
in addition, enables one to precisely describe the intersection of these
classes. In this sense, it is a continuation of \cite{AT10}, where all
self-adjoint (Hermitian)  and at the same time $\PT$ symmetric
operators in $L^2(\dR)$ associated with $\tau_{\epsilon}$
were described.

Surprisingly, it turns out that with the
differential expression $\tau_{q}$ in (\ref{FoundationNew}) there
are $\PT$ symmetric operators which correspond to
one- and three-dimensional extensions of the minimal operator which
are neither Hermitian nor $\mathcal P$-self-adjoint and which
\begin{center}
{\it possesses an empty resolvent set.}
\end{center}

In a next step we will consider complex
deformations, which are, from the mathematical point of view,
less understood. These questions will be treated in a subsequent note.
However, in our opinion even the most "simple" case (i.e.\ $\epsilon =4n+2$,
$x$ real) contains enough unsolved
questions and possesses a rich structure which one needs to understand first.

This paper is organized as follows: After introducing the basic notions
like minimal/maximal operator associated with  $\tau_{q}$
and bi-extensions in Section \ref{One},
we consider $2$-dimensional extensions in Section \ref{Two},  $3$-dimensional extensions in Section \ref{Three} and
 $1$-dimensional extensions in Section \ref{Four}.
 In the case of  $2$-dimensional extensions in Section \ref{Two}
we describe all bi-extensions  which are $\PT$-symmetric or $\cP$-self-adjoint. In the case of $3$-dimensional extensions and
$1$-dimensional extensions there are no $\cP$-self-adjoint nor Hermitian
extensions, but there exists $\PT$-symmetric extensions with empty resolvent
set, cf.\ Sections \ref{Three} and  \ref{Four}.

\section{Preliminaries: Operators in Krein spaces and bi-extensions} \label{One}

Recall that a complex linear space $\mathcal H$ with a hermitian
nondegenerate sesquilinear form $\Skindef$ is called a {\em Krein
space} if there exists a
 so called {\it fundamental decomposition} (cf.\ \cite{AI,B,Krein})
\begin{equation}\label{decomp}
 \cH = \cH_+ \oplus \cH_-
\end{equation}
with subspaces $\cH_\pm$ being orthogonal to each other with respect
to $\Skindef$ such that $(\cH_\pm, \pm\Skindef)$ are Hilbert spaces.
Then
\begin{equation*}\label{decompHilbert}
(x,x):=[x_+,x_+] -  [x_-,x_-], \quad x=x_++x_- \in\mathcal H \quad
\mbox{with } x_\pm \in \cH_\pm,
\end{equation*}
 is an inner product and $({\mathcal H},\Skdef)$ is a Hilbert space.
All topological notions are understood with respect to some Hilbert
space norm $\Norm$ on ${\mathcal H}$ such that $\Skindef$ is
$\Norm$-continuous. Any two such norms are equivalent, see
\cite[Proposition I.1.2]{L}.
 Denote by $P_+$ and $P_-$ the
orthogonal projections onto $\cH_+$ and $\cH_-$, respectively. The
operator $J:= P_+-P_-$ is called the {\it fundamental symmetry}
corresponding to the decomposition (\ref{decomp}) and we have
$$
[f,g] = (Jf,g)\quad\text{for all }f,g\in\mathcal H.
$$
For a detailed treatment of Krein spaces and operators therein we refer to the monographs \cite{AI} and \cite{B}.
If ${\cal L}$ is an arbitrary subset of a Krein space $({\cal H},\Skindef)$
we set
\begin{displaymath}
{\cal L}^{[\perp]} := \{x \in {\cal H} : [x,y] = 0
 \mbox{ for all } y \in {\cal L} \}.
\end{displaymath}
In the sequel we will make use of the following proposition.
\begin{proposition}\label{ort}
Let $\mathcal L$, $\mathcal M$ be  closed subspaces of a Krein space $({\cal H},\Skindef)$ and let $\mathcal L\subset \mathcal M$. Then
$\dim\mathcal L^{[\perp]}/\mathcal M^{[\perp]}=\dim \mathcal M/\mathcal L$.
\end{proposition}
\begin{proof}
Let $J$ be a canonical symmetry in the Krein space $\mathcal K$.
For subspaces $X,Y,Z$ of ${\cal H}$ with $X+Z=Y$ and $X\cap Z =\{0\}$ we obtain
$JX+JZ=JY$ and therefore
$$
\dim Y/X = \dim Z = \dim JZ= \dim JY/JX.
$$
Set $Y:=\mathcal L^{[\perp]}$ and $X:=\mathcal M^{[\perp]}$ we see
$$
\dim \mathcal L^{[\perp]}/\mathcal M^{[\perp]}= \dim J\mathcal L^{[\perp]}/J\mathcal M^{[\perp]}.
$$
As for each subspace $\mathcal N$ the equality  $J\mathcal N^{[\perp]}=\mathcal N^\perp$ holds,
we conclude
 $$
\dim \mathcal L^{[\perp]}/\mathcal M^{[\perp]}= \dim \mathcal L^{\perp}/\mathcal M^{\perp}
=\dim \mathcal M/\mathcal L.
$$
\end{proof}

Let $T$ be a densely defined linear operator in $\mathcal H$. The adjoint of $T$ in the Krein space $(\mathcal H,\Skindef)$ is defined by
\begin{equation}\label{Daissis}
T^+ := JT^*J,
\end{equation}
where $T^*$ denotes the adjoint of $T$ in the Hilbert space  $(\mathcal H,\Skdef)$. We have
$$
[Tf,g] = [f,T^+g]\quad\text{for all }f\in\dom T,\,g\in\dom T^+.
$$
The operator $T$ is called {\it selfadjoint} (in the Krein space  $(\mathcal H,\Skindef)$ ) if $T = T^+$.

\ \\
In what follows, we will consider extensions of a closed densely defined symmetric operator in a Hilbert space $\mathcal H$. As we will consider also non-symmetric extensions,
we will emphazise this in the following definition.
\begin{definition}
A closed extension $\widetilde{A}$ of a closed densely defined symmetric operator $A$ in a Hilbert space $\mathcal H$ is called a bi-extension if
$$
A\subset \widetilde{A}\subset A^*.
$$
For  $r\in \mathbb N$ a bi-extension $\widetilde{A}$ is called
a $r$-dimensional bi-extension, if
$$
\dim (\dom \widetilde{A}/ \dom A )= r.
$$
\end{definition}
For a bi-extension $\widetilde{A}$ of $A$ we have
 $$
 A\subset \widetilde{A}^*\subset A^*.
 $$
Hence both $\widetilde{A}$ and  $\widetilde{A}^*$ are extensions of $A$.

\begin{proposition}\label{ext}
Let $A$ be a closed densely defined symmetric operator in a Hilbert space $\mathcal H$ with the defect indices $(m,n)$ and $p=m+n<\infty$. Then $\widetilde{A}$ is an $r$-dimensional bi-extension of $A$ if and only if $\widetilde{A}^*$ is a $(p-r)$-dimensional bi-extension of $\widetilde{A}$.
\end{proposition}
\begin{proof}
Let us consider the space $\mathcal K:=\mathcal H\times \mathcal H$ as a Krein space with the indefinite metric
$$
[x,y]=\frac{(x_2,y_1)-(x_1,y_2)}{2i},\quad x=
\begin{pmatrix}
 x_1\\ x_2
\end{pmatrix}
,\ y=
\begin{pmatrix}
 y_1\\ y_2
\end{pmatrix}
\in \mathcal H\times\mathcal H.
$$
Hence the symmetry of $A$ implies that the graph $\Gamma_A$ of $A$  is a neutral subspace in $\mathcal K$. Moreover $\Gamma_{A^*}=(\Gamma_A)^{[\perp]}$. The assumption that $\widetilde{A}$ is an  $r$-dimensional bi-extension of $A$ is equivalent to $\Gamma_A\subset \Gamma_{\widetilde{A}}\subset\Gamma_{A^*}$ and $\dim{\Gamma_{\widetilde{A}}}/\Gamma_A=r$. By Proposition \ref{ort} with $\mathcal L=\Gamma_A$ and $\mathcal M=\Gamma_{\widetilde{A}}$ we obtain that $\dim\Gamma_{A^*}/\Gamma_{\widetilde{A}^*}=r$ and therefore from
$\Gamma_A\subset \Gamma_{\widetilde{A}^*}\subset\Gamma_{A^*}$ it follows that $\widetilde{A}^*$ is a $(p-r)$-dimensional bi-extension of $A$.

\end{proof}

\begin{remark}{\rm
 If $A$ is a closed densely defined symmetric operator in a Hilbert space $\mathcal H$ with defect indices $(2,2)$,  then $\widetilde{A}$ is a $1$-dimensional bi-extension of $A$ if and only if $\widetilde{A}^*$ is a $3$-dimensional bi-extension of $A$  and    $\widetilde{A}$  is a $2$-dimensional extension of $A$ if and only if  $\widetilde{A}^*$ is also a $2$-dimensional bi-extension of $A$ .
 }
 \end{remark}

\section{The  Hamiltonian $\tau_{q}$} \label{OneOne}

By $L^2(\dR)$ we denote the space of all equivalence
classes of complex valued,
measurable functions $f$ defined on $\dR$ for which $\int_\dR\vert
f(x)\vert^2 dx$ is finite. We equip $L^2(\dR)$ with the
 usual Hilbert scalar product
\begin{equation*}
(f,g):=\int_\dR f(x)\overline{g(x)}\,dx,\quad f,g\in L^2(\dR).
\end{equation*}
Let $\cP$ represents parity reflection and $\cT$ represents
time reversal as in \eqref{PTDef}.
Then $\cP^2 = \cT^2 = (\PT)^2= I$ and $\PT = \cT\cP$. Observe that the operator $\cT$ is nonlinear.
The operator $\cP$  gives in a natural
way rise to an indefinite inner product $\Skindef$ which will play
an important role in the following. We equip $L^2(\dR)$ with the
indefinite inner product
\begin{equation}\label{Pindef}
[f,g]:=\int_\dR f(x)\overline{(\cP g)(x)}\,dx = \int_\dR
f(x)\overline{g(-x)}\,dx, \quad f,g\in L^2(\dR).
\end{equation}
With respect to this inner product, $L^2(\dR)$ becomes a Krein
space. Observe that in this case the operator $\cP$ serves as a
fundamental symmetry in the Krein space $(L^2(\dR), \Skindef)$. In
the situation where $\Skindef$ is given as in (\ref{Pindef}), it is
easy to see that  the set of all even functions can be chosen
as the positive component $\cH_+$ and the set of all odd functions
can be chosen as the negative component $\cH_-$  in a
decomposition (\ref{decomp}).
We easily see that the $\cP$-self-adjointness from
Definition \ref{PTDEF} coincides with self-adjointness in the Krein
space $(L^2(\mathbb R)\Skindef)$, see \eqref{Daissis}.

\begin{lemma}\label{la}
Let $\widetilde{A}$ be a bi-extension of a
 closed densely defined symmetric operator $A$ in $L^2(\dR)$
  and let $A^*$ be a $\PT$ symmetric operator.
 Then $\widetilde{A}$ is  $\PT$-symmetric if and only if $\widetilde{A}^*$ is $\PT$ symmetric.
\end{lemma}

\begin{proof}
Let $\widetilde{A}$ be  $\PT$-symmetric.
We will show that $\PT \dom\widetilde{A}=\dom\widetilde{A}$ implies
$\PT \dom\widetilde{A}^*=\dom\widetilde{A}^*$ and $\PT \widetilde{A}^*f
=\widetilde{A}^*\PT f$ for all $f\in \dom\widetilde{A}^*$.

Let us note that $f\in\dom\widetilde{A}^*$ if and only if $(\widetilde{A}g,f)=(g,A^*f)$ for all $g\in\dom\widetilde{A}$. For   $g\in\dom\widetilde{A}$ and $f\in\dom\widetilde{A}^*$
the $\PT$ symmetry of $A^*$ implies
\begin{eqnarray*}
(\widetilde{A} g, \PT f)  =(g, A^* \PT f)=(g, \PT A^* f)= (g, \PT \widetilde{A}^* f).
\end{eqnarray*}
From this we conclude $\PT f \in \dom\widetilde{A}^*$ and $\PT \widetilde{A}^* f =
A^* \PT f = \widetilde{A}^*\PT f$. Hence, $\PT\dom\widetilde{A}^* \subset \dom\widetilde{A}^*$
and from $(\PT)^2=I$ we derive $\PT \dom\widetilde{A}^*=\dom\widetilde{A}^*$
and the operator $\widetilde{A}^*$ is $\PT$ symmetric.

If $\widetilde{A}^*$ is $\PT$ symmetric, then, by the first part of the proof,
${\widetilde{A}}^{**} =\widetilde{A}$ is also $\PT$ symmetric.
\end{proof}

\begin{corollary}
  Let $\widetilde{A}$ be a bi-extension of $A$ and let $A^*$ be a $\PT$ symmetric operator. Then $\widetilde{A}$ is  $\PT$-symmetric if and only if $\widetilde{A}^+$ is $\PT$ symmetric.
\end{corollary}
\begin{proof}
 Assume $\widetilde{A}$ is a $\PT$ symmetry. Then  Lemma \ref{la} implies $\PT\dom\widetilde{A}^*=\dom\widetilde{A}^*$.  Since $\PT=\TP$ and $\dom\widetilde{A}^+=\Pt\dom\widetilde{A}^*$ we have
\begin{eqnarray*}
\begin{split}
 \PT \dom\widetilde{A}^+&=\PT\Pt \dom\widetilde{A}^*=\Pt \PT \dom\widetilde{A}^*=\Pt \dom\widetilde{A}^*\\
 &=\dom\widetilde{A}^+,
 \end{split}
 \end{eqnarray*}
what is equivalent to $\PT$ symmetry of $\widetilde{A}^+$.

If $\widetilde{A}^+$ is $\PT$ symmetric, then, by the first part of the proof,
${\widetilde{A}}^{++} =\widetilde{A}$ is also $\PT$ symmetric.
\end{proof}

In the following, we
consider the differential expression $\tau_{q}$.
We assume that $q$ is a real valued function from
 $L_{loc}^1(\mathbb R)$ which is even, that is,
 $$
 q(x)=q(-x)\quad  \mbox{for all }
 x\in \mathbb R,
 $$
  such that the differential equation
\begin{equation}\label{FoundationNew2}
\tau_{q} (y) (x) := - y^{\prime \prime}(x) - q(x)
 y(x), \quad x\in \mathbb R.
\end{equation}
is in limit  circle case at $+\infty$ and
$-\infty$.

From \cite[Remark 7.4.2 (2)]{Z05}\footnote{In the formulation of
\cite[Example 7.4.1]{Z05} and, hence, in \cite[Remark 7.4.2
(2)]{Z05}  a minus sign is missing.} we see that, e.g., $\tau_{q}$
is in limit  circle case at $+\infty$ and
$-\infty$ for all $\delta >0$ and for
$$
q(t) = t^{2+\delta}.
$$
 Hence, the
 differential expression $\tau_{4n+2}$ in \eqref{EBurg2} is
in the limit circle case at $\infty$ and at $-\infty$.

Recall that $\tau_{q}$ is called in limit circle at $\infty$ (at
$-\infty$) if all solutions of the equation $\tau_{q}(y) -
\lambda y=0$, $\lambda \in \mathbb C$, are in $L^2((a,\infty))$
(resp.\ $L^2((-\infty,a))$) for some $a\in \mathbb R$, cf.\
e.g.\ \cite[Chapter 7]{Z05}, \cite[Section 5]{W} or \cite[Section 13.3]{W03}.

With the differential expression $\tau_{q}$ we will associate an operator $A_{\mmm}$
defined on the maximal domain $\cD_{\mmm}$, i.e.,
$$
\cD_{\mmm}:= \{y \in L^2(\dR) : y,y^\prime \in AC_{loc}(\dR),
 \tau_{q}y \in L^2(\dR)\},
$$
via
$$
\dom A_\mmm := \cD_{\mmm}, \quad  A_\mmm y := \tau_{q}(y) \quad \mbox{for } y
\in \cD_{\mmm}.
$$
Here and in the following $AC_{loc}(\dR)$ denotes the space of all
complex valued functions which are
 absolutely continuous on all compact
 subsets of $\dR$. As usual, with \eqref{FoundationNew2}, there is also connected
 the so-called {\it pre-minimal} operator $A_0$ defined on the  domain
 $$
\cD_{0}:= \{y \in  \cD_{\mmm} : y \mbox{ has compact support} \}
$$
and defined via
$$
\dom A_0 := \cD_{0}, \quad  A_0 y := \tau_{q}(y) \quad \mbox{for } y
\in \cD_{0}.
$$
The operator $A_0$ is symmetric and not closed but it is closable. Its closure
$\overline{A_0}$ is called the {\it minimal} operator and we denote
it by $A$,
$$
A:= \overline{A_0}.
$$
It turns out that the maximal operator is precisely the adjoint of the minimal
operator, see, e.g., \cite[Theorem 3.9]{W} or \cite[Lemma 10.3.1]{Z05},
$$
A^* =(\overline{A_0})^*= A_\mmm .
$$
Obviously, by the definition of the maximal and the minimal operator the
following lemma holds true.
\begin{lemma}\label{Jetztdoch}
The operators $A$ and $A^*$ are $\PT$-symmetric.
\end{lemma}

Moreover, by \cite[Theorem 10.4.1]{Z05} or \cite[Theorem 5.7]{W}, we obtain a
statement on the deficiency indices of $A$.

\begin{lemma} \label{ArnstadtSued}
The closed symmetric operator $A$ has deficiency indices $(2,2)$, i.e.\
$\dim \ker (A^*-i)= 2= \dim \ker (A^*+i)$. In particular, we have
$$
\dim\ \dom A^*/\dom{A}=4.
$$
\end{lemma}

Hence, by Lemma \ref{ArnstadtSued}, bi-extensions $\widetilde{A}$ of $A$
are either trivial, that is, they equal $A$ or $A^*$ or they fell into one
of the following cases
\begin{itemize}
  \item  $\dim \dom \widetilde{A}/\dom{A}=1$; this case is discussed in Section \ref{Four},
  \item  $\dim \dom\widetilde{A}/\dom{A}=2$; this case is discussed in Section \ref{Two},
  \item  $\dim \dom\widetilde{A}/\dom{A}=3$; this case is discussed in Section \ref{Three}.
\end{itemize}

It is our aim to describe all bi-extensions $\widetilde{A}$ of $A$.
For this we define for functions
$g,f \in AC_{loc}(\dR)$ with continuous derivative, the expression
$[f,g]_x$ for $x \in \dR$ via
$$
[f,g]_x := \overline{f(x)}g^\prime(x) - \overline{f^\prime(x)}g(x).
$$
Note that if $f$ and $g$ are real valued, then $[f,g]_x$ is the
Wronskian $W(f,g)$. It is well known (e.g.\ \cite[Lemma 10.2.3]{Z05},
\cite[Theorem 3.10]{W})
that the limit of $[g,f]_x$ as
$x\to \infty$ and $x\to -\infty$ exists for $f,g \in \cD_{\mmm}$,
  We set
$$
[f,g]_\infty := \lim_{x \to \infty}[f,g]_x, \quad
[f,g]_{-\infty} := \lim_{x \to -\infty}[f,g]_x
$$
and
$$
[f,g]_{-\infty}^\infty=[f,g]_{\infty}-[f,g]_{-\infty}.
$$
By \cite[Section 10.4.4]{Z05}, \cite[Theorem 3.10]{W}, we have for $f,g \in \cD_{\mmm}$
\begin{equation}\label{Wolfsburg}
(g,A^*f)-(A^*g,f) = [f,g]_{-\infty}^\infty.
\end{equation}

The following Lemmas \ref{odd} and \ref{alpha_1} are from \cite{AT10}, where
they are proved for the differential expression $\tau_{4n+2}$. However, it is easy to see that  the proofs in \cite{AT10} also applies to the differential expression $\tau_{q}$ due to the assumption that $q$ is an even function.
\begin{lemma}\label{odd}
 There exist real valued solutions $w_1,w_2 \in \cD_{\mmm}$ of
the equation
$$
\tau_{q}(y) =0
$$
such that $w_1$ is an odd and $w_2$ an even function with
$$
[w_1,w_2]_{-\infty}= [w_1,w_2]_\infty =1
$$
and
$$
[w_1,w_1]_{-\infty}= [w_1,w_1]_\infty =[w_2,w_2]_{-\infty}=
[w_2,w_2]_\infty =0.
$$
\end{lemma}
  For simplicity we set for $f \in \cD_{\mmm}$
$$
\begin{array}{ll}
 \alpha_1(f):= [w_1,f]_{-\infty},\quad & \quad
 \alpha_2(f):= [w_2,f]_{-\infty}, \\
\beta_1(f):= [w_1,f]_{\infty}, \quad & \quad \beta_2(f):=
[w_2,f]_{\infty}.
\end{array}
$$
We obtain (see, e.g.\ \cite[Satz 13.21]{W03})
\begin{equation}\label{Sevilla}
\dom A = \{ f\in  \cD_{\mmm} : \alpha_1(f) =\alpha_2(f)=
\beta_1(f)=\beta_2(f)=0 \}.
\end{equation}

\begin{lemma}\label{HornoEnSeville}
To each vector $z=(z_1,z_2,z_3,z,_4)^\top$ in $\mathbb C^4$ there exists a
function $f_z$ from the domain $ \cD_{\mmm}$ of the maximal operator $A^*$
with
\begin{eqnarray*}
\begin{split} \alpha_1(f_z)&= z_1,\\
\beta_1(f_z)&= z_3,
\end{split}
\quad & \quad
\begin{split}
\alpha_2(f_z)&= z_2 , \\
 \beta_2(f_z)&=z_4.
\end{split}
\end{eqnarray*}
\end{lemma}
\begin{proof}
We consider functions $u_1,u_2,v_1,v_2$ from $\cD_{\mmm}$ such that
$u_j$, $j=1,2$ equal $w_j$ on the interval $(1,\infty)$, equal zero
on the interval $(-\infty,-1)$ and the functions $v_j$, $j=1,2$
equal $w_j$ on the interval $(-\infty,-1)$ and equal zero on the
interval $(1,\infty)$. Then
$$
f_z:= -z_4 u_1 + z_3 u_2 - z_2 v_1 + z_1
v_2.
$$
is the function with the desired properties.
\end{proof}

The next lemma\footnote{Here we mention that in the second part of the statement of \cite[Lemma 4]{AT10} $\cT$ should be replaced by $\PT$.} describes the behaviour of the above numbers under
the operators $\cP$ and $\cT$.

\begin{lemma}\label{alpha_1}
For $f \in \cD_{\mmm}$ we have
\begin{eqnarray}
\begin{split}
\begin{split} \alpha_1(\cP f)&= \beta_1(f),\\
\beta_1(\cP f)&= \alpha_1(f),
\end{split}
\quad & \quad
\begin{split}
\alpha_2(\cP f)&= -\beta_2(f) , \\
 \beta_2(\cP f)&=
-\alpha_2(f),
\end{split}
\end{split} \label{p}\\[0.5cm]
\begin{split}
 \alpha_1(\PT f)= \overline{\beta_1(f)},\quad & \quad
\alpha_2(\PT f)= -\overline{\beta_2(f)} , \\
\beta_1(\PT f)= \overline{\alpha_1(f)}, \quad & \quad
\beta_2(\PT
f)= -\overline{\alpha_2(f)}.
\end{split}\label{pt}
\end{eqnarray}
\end{lemma}

\section{$2$-dimensional extensions} \label{Two}
First we will consider $2$-dimensional extensions $\widetilde{A}$.
Their domain is given by
\begin{equation}\label{2dim}
\dom\widetilde{A}=\left\{f\in \cD_\mmm \mid \begin{bmatrix}
  a_1&b_1\\
  c_1&d_1
\end{bmatrix}
\begin{pmatrix}
 \alpha_1(f)\\ \alpha_2(f)
\end{pmatrix}  =\begin{bmatrix}
  e&f\\
  g&h
\end{bmatrix}
\begin{pmatrix}
 \beta_1(f)\\ \beta_2(f)
\end{pmatrix}  \right\}
\end{equation}
with
$$
\rm rank\ \begin{bmatrix}
  a_1&b_1&e&f\\
  c_1&d_1&g&h
\end{bmatrix}=2.
$$
There are 3 possibilities:
\begin{enumerate}
 \item[$(i)$] The matrix $\begin{bmatrix}
  a_1&b_1\\
  c_1&d_1
\end{bmatrix}$ is nondegenerate. Then we can express
$\alpha_1(f),\ \alpha_2(f)$ via $\beta_1(f),\ \beta_2(f)$:
\begin{equation}\label{ab}
\begin{pmatrix}
 \alpha_1(f)\\ \alpha_2(f)
\end{pmatrix}  =\begin{bmatrix}
  a&b\\
  c&d
\end{bmatrix}
\begin{pmatrix}
 \beta_1(f)\\ \beta_2(f)
\end{pmatrix},
\end{equation}
where
$$
\begin{bmatrix}
  a&b\\
  c&d
\end{bmatrix}=\begin{bmatrix}
  a_1&b_1\\
  c_1&d_1
\end{bmatrix}^{-1}\, \begin{bmatrix}
  e&f\\
  g&h
\end{bmatrix}.
$$
Hence
\begin{equation}\label{mix1}
\dom\widetilde{A}=\left\{
f\in\cD_{\mmm}\mid \begin{pmatrix}
 \alpha_1(f)\\ \alpha_2(f)
\end{pmatrix}  =\begin{bmatrix}
  a&b\\
  c&d
\end{bmatrix}
\begin{pmatrix}
 \beta_1(f)\\ \beta_2(f)
\end{pmatrix}
\right\} .
\end{equation}

 \item[$(ii)$] The matrix $\begin{bmatrix}
  e&f\\
  g&h
\end{bmatrix}$ is nondegenerate. Then we can express
$\beta_1(f),\ \beta_2(f)$ via $\alpha_1(f),\ \alpha_2(f)$ and rewrite \eqref{2dim} in the form:
\begin{equation}\label{ba}
\begin{pmatrix}
 \beta_1(f)\\ \beta_2(f)
\end{pmatrix}
  =\begin{bmatrix}
  a&b\\
  c&d
\end{bmatrix}
\begin{pmatrix}
 \alpha_1(f)\\ \alpha_2(f)
\end{pmatrix},
\end{equation}
where in this case
$$
\begin{bmatrix}
  a&b\\
  c&d
\end{bmatrix}=\begin{bmatrix}
  e&f\\
  g&h
\end{bmatrix}^{-1}\,\begin{bmatrix}
  a_1&b_1\\
  c_1&d_1
\end{bmatrix}.
$$
Hence
\begin{equation}\label{mix2}
\dom\widetilde{A}=\left\{
f\in\cD_{\mmm}\mid \begin{pmatrix}
 \beta_1(f)\\ \beta_2(f)
\end{pmatrix}
  =\begin{bmatrix}
  a&b\\
  c&d
\end{bmatrix}
\begin{pmatrix}
 \alpha_1(f)\\ \alpha_2(f)
\end{pmatrix}
\right\} .
\end{equation}
  \item[$(iii)$] Both matrices $\begin{bmatrix}
  a_1&b_1\\
  c_1&d_1
\end{bmatrix}$ and $\begin{bmatrix}
  e&f\\
  g&h
\end{bmatrix}$ are degenerate. Then they are both of rank = 1 and therefore there exist  numbers $a ,\, b,\, c,\, d$ with $|a|+|b |\neq 0$ and $|c|+|d|\neq 0$ such that one can rewrite \eqref{2dim} as a system:
\begin{equation}\label{abc}
\left\{
\begin{array}{lcr}
 a\,\alpha_1(f)+b\,\alpha_2(f)&=&0,\\
 c\, \beta_1(f)+d\, \beta_2 (f)&=&0.
\end{array}
\right.
\end{equation}
\end{enumerate}

Let us recall that \eqref{ab} and \eqref{ba} are called mixed boundary conditions and \eqref{abc} is called separated boundary conditions.
By our assumptions, cases $(i)$ and $(iii)$ can not occur simultaneously.
Similarly,  cases $(ii)$ and $(iii)$ can not occur simultaneously.

We normalize \eqref{abc} and rewrite for this case \eqref{2dim} as
\begin{equation}\label{separate}
 \dom \widetilde{A}=\left\{f\in \cD_\mmm\mid   \begin{array}{lcr}
  \alpha_1(f)\,\xi \cos \alpha - \alpha_2(f)\,\sin \alpha&=&0,\\
   \beta_1(f)\, \eta \cos \beta -\beta_2 (f)\,  \sin \beta&=&0.
  \end{array}\right\}
  \end{equation}
  Here $|\xi|=|\eta| =1$ and $\alpha,\, \beta\,\in [0,2\pi)$.

Note that in the case of separated boundary conditions there exist vectors $f_1,f_2\in\dom\widetilde{A}$ such that:
\begin{equation}\label{nonzero}
|\alpha_1(f_1)|\,+\,|\alpha_2(f_1)|\neq 0,\qquad  |\beta_1(f_2)|\,+\,|\beta_2(f_2)|\neq 0.
\end{equation}
which is due to the fact that $\widetilde{A}$ is a 2-dimensional
extension of $A$.

In both cases, separated and mixed boundary conditions, extensions $\widetilde{A}$ are bi-extensions since
$A\subset \widetilde{A}\subset A^*$. Our aim in this section is to describe adjoint and $\cP$-adjoint operators to the  extension $\widetilde{A}$ and give criteria when this extension is $\PT$ symmetry, selfadjoint and $\Pt$-selfadjoint.

For this we need the following  result.
\begin{lemma}\label{lab}
Let $f,g\in \cD_\mmm$. Then
\begin{equation}\label{eqfg}
[g,f]_{-\infty}^\infty=\overline{\beta_2(g)}\beta_1(f)-
\overline{\alpha_2(g)}\alpha_1(f)-
\overline{\beta_1(g)}\beta_2(f)+\overline{\alpha_1(g)}\alpha_2(f).
\end{equation}
\end{lemma}
\begin{proof}
Consider the function
\begin{equation}\label{eF}
F(x;g,f,w_1,w_2)=[g,f]_x[w_1,w_2]_x.
\end{equation}
 A direct calculation shows that
\begin{equation}\label{eFF}
F(x;g,f,w_1,w_2)=[g,w_2]_x[w_1,f]_x-[g,w_1]_x[w_2,f]_x.
\end{equation}
Since $[w_1,w_2]_{-\infty}=[w_1,w_2]_\infty=1$, on the one hand from \eqref{eF} it follows that
$$
\lim\limits_{x\to \infty}F(x;g,f,w_1,w_2)-\lim\limits_{x\to -\infty}F(x;g,f,w_1,w_2)=[g,f]_{-\infty}^\infty,
$$
and from the other hand side \eqref{eFF} implies
$$
\lim\limits_{x\to \infty}F(x;g,f,w_1,w_2)-\lim\limits_{x\to -\infty}F(x;g,f,w_1,w_2)=
$$
$$
\overline{\beta_2(g)}\beta_1(f)-
\overline{\alpha_2(g)}\alpha_1(f)-
\overline{\beta_1(g)}\beta_2(f)+\overline{\alpha_1(g)}\alpha_2(f) .
$$
Therefore  \eqref{eqfg} holds.
\end{proof}
\begin{corollary}\label{c1}
 Let $\widetilde{A}$ be a bi-extension of $A$. Then $g\in \dom\widetilde{A}^*$ if and only if
 \begin{equation}\label{w}
  \overline{\beta_1(g)}\beta_2(f)-
\overline{\beta_2(g)}\beta_1(f)=\overline{\alpha_1(g)}\alpha_2(f)
-\overline{\alpha_2(g)}\alpha_1(f)
 \end{equation}
 for all $f\in\dom\widetilde{A}$.
\end{corollary}
\begin{proof}
This follows immediately from Lemma \ref{lab} and the fact that  $g\in \dom\widetilde{A}^*$ if and only if $[g,f]_{-\infty}^\infty=0$ for all $f\in\dom\widetilde{A}$, see \eqref{Wolfsburg}.
\end{proof}
\begin{proposition}\label{domT*}
 $(1)$\ If $\widetilde{A}$ is given by \eqref{separate}, then
 \begin{equation}\label{domT*sep}
  \dom  \widetilde{A}^*=\left\{g\in \cD_\mmm\mid   \begin{array}{lcr}
  \alpha_1(g) \cos \alpha - \alpha_2(g)\xi\sin \alpha&=&0,\\
   \beta_1(g)   \cos \beta -\beta_2 (g) \eta \sin \beta&=&0,
  \end{array}\right\}
 \end{equation}

\begin{equation}\label{domT+sep}
  \dom  \widetilde{A}^+=\left\{g\in \cD_\mmm\mid   \begin{array}{lcr}
  \alpha_1(g) \cos \beta + \alpha_2(g)\eta\sin \beta&=&0,\\
   \beta_1(g)   \cos \alpha +\beta_2 (g) \xi \sin \alpha&=&0.
  \end{array}\right\}
 \end{equation}

 $(2)$\ \ If $\widetilde{A}$ is given by \eqref{mix2}, then
 \begin{equation}\label{domT*mixed}
   \dom  \widetilde{A}^*=\left\{g\in \cD_\mmm\mid \begin{pmatrix}
     \alpha_1(g)\\ \alpha_2(g)
   \end{pmatrix}=\begin{bmatrix}
    \overline{d}& -\overline{b}\\
    -\overline{c}&\overline{a}
   \end{bmatrix}\begin{pmatrix}
     \beta_1(g)\\
     \beta_2(g)
   \end{pmatrix}\right\},
  \end{equation}

  \begin{equation}\label{domT+mixed}
\dom  \widetilde{A}^+=\left\{g\in \cD_\mmm\mid \begin{pmatrix}
     \beta_1(g)\\ \beta_2(g)
   \end{pmatrix}=\begin{bmatrix}
    \overline{d}& \overline{b}\\
    \overline{c}&\overline{a}
   \end{bmatrix}\begin{pmatrix}
     \alpha_1(g)\\
     \alpha_2(g)
   \end{pmatrix}\right\}.
  \end{equation}

$(3)$\ If $\widetilde{A}$ satisfy \eqref{mix1} then
\begin{equation}\label{domT*mixed+}
   \dom  \widetilde{A}^*=\left\{g\in \cD_\mmm\mid
   \begin{pmatrix}
     \beta_1(g)\\
     \beta_2(g)
   \end{pmatrix}
   =\begin{bmatrix}
    \overline{d}& -\overline{b}\\
    -\overline{c}&\overline{a}
   \end{bmatrix}\right\}\begin{pmatrix}
     \alpha_1(g)\\ \alpha_2(g)
   \end{pmatrix},
  \end{equation}
\begin{equation}\label{domT+mixed+}
\dom  \widetilde{A}^+=\left\{g\in \cD_\mmm\mid \begin{pmatrix}
     \alpha_1(g)\\
     \alpha_2(g)
   \end{pmatrix}=\begin{bmatrix}
    \overline{d}& \overline{b}\\
    \overline{c}&\overline{a}
   \end{bmatrix}\begin{pmatrix}
     \beta_1(g)\\ \beta_2(g)
   \end{pmatrix}
\right\}.
  \end{equation}
\end{proposition}
\begin{proof}
Since  $\widetilde{A}$ and operators with domains \eqref{domT*sep}, \eqref{domT*mixed}, and \eqref{domT*mixed+}
correspond to two dimensional extensions of $A$ for a proof of the statement it is sufficient to check \eqref{w}   for $f\in\dom \widetilde{A}$ and $g$ from \eqref{domT*sep}, \eqref{domT*mixed}, or \eqref{domT*mixed+}, respectively. But this directly follows from \eqref{eqfg}.\\

Since $\dom\widetilde{A}^+=\{g=\Pt f\mid f\in \dom\widetilde{A}^*\}$ a proof of \eqref{domT+sep}, \eqref{domT+mixed}, and \eqref{domT+mixed+}
taking in account Lemma \ref{alpha_1}, relations \eqref{p}, \ follows immediately from \eqref{domT*sep} and \eqref{domT*mixed}, or \eqref{domT*mixed+}, respectively.
 \end{proof}
\begin{corollary}
$(1)$\ Let $\widetilde{A}$ has the domain \eqref{separate}. Then
$\widetilde{A}=\widetilde{A}^*$ if and only if  $\xi=\eta=1$.\\

$(2)$\ Let $\widetilde{A}$ has the domain \eqref{mix1} or \eqref{mix2}.
Set $\Delta=ad-bc$. Then
 $\widetilde{A}=\widetilde{A}^*$ if and only if for some
 $\varphi \in \mathbb R$ we have
  \begin{equation}\label{a=a*}
 \Delta= e^{2i\varphi} \mbox{ and }\, \left\{
 e^{-i\varphi}a,\, e^{-i\varphi}b,\, e^{-i\varphi}c,\, e^{-i\varphi}d
 \right\}
 \subset \mathbb R.
     \end{equation}
\end{corollary}
\begin{proof}
Let us note that $\widetilde{A}=\widetilde{A}^*$ if and only if
$\dom\widetilde{A}=\dom\widetilde{A}^*$. Assertion
$(1)$\ follows immediately if one compares $\dom\widetilde{A}$ and
$\dom\widetilde{A}^*$.\\

$(2)$\ Assume $\widetilde{A}$ has domain \eqref{mix1}.
Obviously, if \eqref{a=a*} holds, then we have
 $$
\begin{bmatrix}
   a& b\\
   c&d
   \end{bmatrix} \begin{bmatrix}
    \overline{d}& -\overline{b}\\
    -\overline{c}&\overline{a}
   \end{bmatrix} =
   \begin{bmatrix}
   1& 0\\
   0&1
   \end{bmatrix}
 $$
 and
 $\widetilde{A}=\widetilde{A}^*$, see \eqref{domT*mixed+}.

For the contrary, we assume $\dom\widetilde{A}=\dom\widetilde{A}^*$.
First, we will show that in this case the matrix
\begin{equation}\label{Palma}
\begin{bmatrix}
    a&b\\
    c&d   \end{bmatrix}
\end{equation}
has full rank. Indeed, assume that there exists $\eta \in \mathbb C$
with $c=\eta a$ and $d = \eta b$. From \eqref{mix1}, \eqref{domT*mixed+}
we obtain for $f\in \dom \widetilde{A}$
$$
\alpha_2(f) =\eta \alpha_1 (f), \quad \beta_1(f) =
\overline{b}(\overline{\eta} -\eta)  \alpha_{1}(f)
\quad  \mbox{and} \quad \beta_2(f) =\overline{a}(\overline{\eta} -\eta)  \alpha_{1}(f).
$$
But this implies that $\widetilde{A}$ is a $1$-dimensional extension
of $A$, a contradiction. Hence, the matrix in \eqref{Palma} has full rank.

There are two vectors $f,g\in\dom\widetilde{A}$ such that the vectors
\begin{equation}\label{Nuernberg222}
\begin{pmatrix}
     {\alpha_1(f)}\\
      {\alpha_2(f)}
\end{pmatrix}\quad \mbox{and} \quad
\begin{pmatrix}{\alpha_1(g)}\\ {\alpha_2(g)}
   \end{pmatrix} \quad \mbox{are linearly independent.}
   \end{equation}
Indeed, assume that all such vectors are linearly dependent. Since
$\dom \widetilde{A}\neq \dom A$ and by \eqref{Sevilla} there is a vector $f_0\in\dom \widetilde{A}$ such that $|\alpha_1(f_0)|+|\alpha_2(f_0)|\neq 0.$ Then for each $f\in\dom\widetilde{A}$ there exists a number $\lambda(f) \in \mathbb C$ with
\begin{equation}\label{eqbeta222}
  \begin{pmatrix}
     {\alpha_1(f)}\\
      {\alpha_2(f)}
\end{pmatrix}=\lambda(f)\,\begin{pmatrix}
     {\alpha_1(f_0)}\\
      {\alpha_2(f_0)}
\end{pmatrix}.
\end{equation}
and, from \eqref{mix1} and the fact that the matrix in \eqref{Palma}
has full rank, we deduce
\begin{equation}\label{eqalpha222}
  \begin{pmatrix}
     {\beta_1(f)}\\
      {\beta_2(f)}
\end{pmatrix}=\lambda(f)\,\begin{pmatrix}
     {\beta_1(f_0)}\\
      {\beta_2(f_0)}
\end{pmatrix}.
\end{equation}
 Using \eqref{eqalpha222} and \eqref{eqbeta222} one can conclude that the  functions $f$ from the  domain of $\widetilde{A}$ satisfy the following system
 $$
 \begin{cases}
   \alpha_1(f)\alpha_2(f_0)-\alpha_2(f)\alpha_1(f_0)=0,\\
   \beta_1(f)\beta_2(f_0)-\beta_2(f)\beta_1(f_0)=0,
 \end{cases}
 $$
 that is, the boundary conditions are separated, a contradiction.
 Hence there are two vectors $f$ and $g$ in $\dom \widetilde{A}$ such
 that \eqref{Nuernberg222} holds.
As $\dom\widetilde{A}=\dom\widetilde{A}^*$,
both boundary conditions \eqref{mix1} and \eqref{domT*mixed+} hold, that is,
for $f\in \dom\widetilde{A}$ we have
$$
  \begin{pmatrix}
 \alpha_1(f)\\ \alpha_2(f)
\end{pmatrix}  =\begin{bmatrix}
  a&b\\
  c&d
\end{bmatrix}
\begin{pmatrix}
 \beta_1(f)\\ \beta_2(f)
\end{pmatrix}=\begin{bmatrix}
  a&b\\
  c&d
\end{bmatrix}\,\begin{bmatrix}
    \overline{d}& -\overline{b}\\
    -\overline{c}&\overline{a}
   \end{bmatrix}\,\begin{pmatrix}
     \alpha_1(f)\\ \alpha_2(f)
   \end{pmatrix}.
$$
The matrix in \eqref{Palma} has full rank  and from \eqref{Nuernberg222} we obtain
$|\Delta|=1$. That is $\Delta=e^{2i\varphi}$, $\varphi=\frac{1}{2}\arg \Delta$. In particular it follows
 $$
  \begin{bmatrix}
a&b\\
c&d
  \end{bmatrix}=\frac{1}{\overline{\Delta}}\,\begin{bmatrix}
    \overline{a}&\overline{b}\\
    \overline{c}&\overline{d}
   \end{bmatrix}.
$$
Hence,
$$
e^{-i\varphi}\,a =e^{-i\varphi}\,
\frac{\overline{a}}{\overline{\Delta}}=
\frac{\overline{a}}{e^{-i\varphi}} =
\overline{e^{-i\varphi}\,a}
$$
and $e^{-i\varphi}\,a$ is real. By similar arguments, one conclude
 $e^{-i\varphi}\,b \in \mathbb R$, $e^{-i\varphi}\,c
 \in \mathbb R$ and $e^{-i\varphi}\,d \in \mathbb R$ and \eqref{a=a*} holds.
The case when $\dom\widetilde{A}$ is defined by \eqref{mix2} can be proved by the same arguments.
\end{proof}

\begin{corollary}
\begin{itemize}
\item[\rm (1)]
Let $\widetilde{A}$ has the domain \eqref{separate}. Then
\begin{enumerate}
\item[$(i)$] for $\alpha\in \left\{0,\,\pi/2,\,\pi,\,3\pi/2\right\} $ or  $\beta\in \left\{0,\,\pi/2,\,\pi,\,3\pi/2\right\} $
the operator $\widetilde{A}$ is $\Pt$-selfadjoint if and only if
$$
\alpha+\beta=0\ (\, \hspace*{-9pt}\mod{\pi}).
$$
\item[$(ii)$] For $\alpha,\beta\notin \left\{0,\,\pi/2,\,\pi,\,3\pi/2\right\} $ the operator $\widetilde{A}$ is $\Pt$-selfadjoint if and only if  one of the following conditions holds:
    \begin{eqnarray}
    \alpha+\beta&=&0\ (\hspace*{-9pt}\mod{\pi}),\ \xi\,\eta=1
    \mbox{ or}\label{1s}\\
    |\alpha-\beta|&=&0\ (\hspace*{-9pt}\mod{\pi}),\ \xi\eta=-1.\label{4s}
    \end{eqnarray}
 \end{enumerate}
\item[\rm (2)]
Let $\widetilde{A}$ has the domain given by \eqref{mix1} or \eqref{mix2}. Then the operator $\widetilde{A}$ is $\Pt$-selfadjoint
   if and only if
\begin{equation}\label{mixp}
d=\overline{a},\quad b,\,c\in\mathbb R.
\end{equation}
\end{itemize}
\end{corollary}
\begin{proof} Since both $\widetilde{A}$ and $\widetilde{A}^+$ are restrictions of the same maximal operator we have
\begin{equation}\label{pdom}
\widetilde{A}=\widetilde{A}^+\quad \Longleftrightarrow\quad \dom\widetilde{A}=\dom\widetilde{A}^+.
\end{equation}

$(1)$\  Assume $\widetilde{A}$ has the domain \eqref{separate}.
We prove the statement $(i)$ for  $\alpha=0$.  For a proof of the other cases,
i.e.,  $\alpha\in \left\{\pi/2,\,\pi,\,3\pi/2\right\} $, $\beta\in \left\{0,\,\pi/2,\,\pi,\,3\pi/2\right\}$ one can use similar arguments.

 Let $\alpha =0$.
 We will show  $A=A^+$ if and only if $\beta=0$ or $\beta=\pi$.

If $\beta=0$ or $\beta=\pi$ then, by \eqref{separate} and \eqref{domT+sep}
we see immediately
 $\dom\widetilde{A}=\dom\widetilde{A}^+$ and according to \eqref{pdom}   $\widetilde{A}$ is $\Pt$-selfadjoint.

 Conversely, assume $\widetilde{A}=\widetilde{A}^+$.
  Then $\alpha =0$ implies $\alpha_1(f)=\beta_1(f)=0$ for
  all $f\in \dom  \widetilde{A} = \dom  \widetilde{A}^+$.
  Hence from \eqref{nonzero} it follows that $\sin\beta=0$, that is, either $\beta=0$ or $\beta=\pi$ and statement $(i)$  is proved.

In order to show $(ii)$ let $\alpha,\beta\notin \left\{0,\,\pi/2,\,\pi,\,3\pi/2\right\}$ and assume $\widetilde{A}$ is $\Pt$-selfadjoint. Then
for $f\in \dom \widetilde{A}=\dom \widetilde{A}^+$
the boundary conditions   \eqref{separate} and \eqref{domT+sep}  give
\begin{eqnarray*}\label{first}
\left\{\begin{split}
 \alpha_1(f)\,\xi \cos \alpha - \alpha_2(f)\,\sin \alpha&=0,\\
   {\alpha_1(f)}\,  \overline{\eta} \cos \beta + {\alpha_2 (f)}\,  \sin \beta&=0,
\end{split}\right.
\end{eqnarray*}
\begin{eqnarray*}\label{second}
\left\{\begin{split}
 \beta_1(f)\, \eta \cos \beta -\beta_2 (f)\,  \sin \beta&=0.\\
{\beta_1(f)}\, \overline{\xi} \cos \alpha +  {\beta_2(f)}\,\sin \alpha&=0.
\end{split}\right.
\end{eqnarray*}
Taking into account \eqref{nonzero} we obtain
$$
\eta \cos\beta\sin\alpha+\overline{\xi}\sin\beta\cos\alpha=0,
$$
or, since $|\xi|=|\eta|=1$,
$$
\xi\eta \cos\beta\sin\alpha+ \sin\beta\cos\alpha=0.
$$
By the conditions $\alpha,\beta\notin \left\{0,\,\pi/2,\,\pi,\,3\pi/2\right\} $ and therefore
 all the numbers $\sin\alpha$, $\cos\alpha$, $\sin\beta$, $\cos\beta$ are nonzero. Therefore $\xi\eta\in\mathbb R$, that is, $ \xi\eta=\pm 1$. Hence
 we have two possibilities:
 $$
 \sin (\alpha-\beta)=0,\quad \mbox{ or } \quad \sin(\alpha+\beta)=0.
 $$
 The latter is equivalent to \eqref{1s}, \eqref{4s}.

 The converse can be checked directly.

 \ \\
$(2)$\ Let $\widetilde{A}$ has the domain \eqref{mix1}. Then it is $\Pt$-selfadjoint if and only if functions $f$ from $\dom \widetilde{A}$
 satisfy also conditions \eqref{domT+mixed+}. Therefore
$$
\begin{bmatrix}
    \overline{d}-a&\overline{b}-b\\
    \overline{c}-c&\overline{a}-d
   \end{bmatrix}\begin{pmatrix}
     {\beta_1(f)}\\
      {\beta_2(f)}
   \end{pmatrix}=0.
$$
Since $\widetilde{A}$ is a $2$-dimensional extension of $A$   \eqref{mixp} holds.

The converse statement can be checked by direct calculations.
A proof for $\widetilde{A}$ with domain \eqref{mix2} is similar.
\end{proof}
\begin{proposition}\label{1412}
\begin{itemize}
\item[\rm (1)] Let $\widetilde{A}$ has the domain  \eqref{separate}.
Then $\widetilde{A}$ is $\PT$ symmetric if and only if it is $\Pt$-selfadjoint.
\item[\rm (2)] Let $\widetilde{A}$ has the domain given by \eqref{mix1}
or \eqref{mix2}. Then
$\widetilde{A}$ is $\PT$ symmetric if and only if for some $\varphi \in \mathbb R$
\begin{equation}\label{mixpt}
\begin{bmatrix}
  a&b\\
  c&d
\end{bmatrix}
=e^{i\varphi}\begin{bmatrix}
    \alpha&\beta\\
    \gamma&\overline{\alpha}
   \end{bmatrix},\quad \beta,\gamma \in\mathbb R,\quad |\alpha|^2-\beta\gamma=1.
\end{equation}
\end{itemize}
\end{proposition}
\begin{proof}
Taking into account that $A^*$ is a $\PT$ symmetry we conclude that $\widetilde{A}$ is a $\PT$ symmetry if and only if
\begin{equation}\label{ptdom}
\PT\dom \widetilde{A}=\dom\widetilde{A}.
\end{equation}
$(1)$\ According to   \eqref{pt} the latter means that
$f \in \dom \widetilde{A}$ also satisfies the conditions
\begin{eqnarray*}
 {\beta_1(f)}\, \overline{\xi} \cos \alpha +  {\beta_2(f)}\,\sin \alpha&=&0,\\
   {\alpha_1(f)}\,  \overline{\eta} \cos \beta + {\alpha_2 (f)}\,  \sin \beta&=&0
  \end{eqnarray*}
which coincide with boundary conditions \eqref{domT+sep} for $\widetilde{A}^+$. Hence statement $(1)$  follows.

$(2)$\ Let boundary conditions of $\widetilde{A}$ be  not separated,  for instance,  let the domain of $\widetilde{A}$ satisfies \eqref{mix1}. According to \eqref{pt} the equality $\PT\dom \widetilde{A}=\dom \widetilde{A}$   is equivalent to \eqref{mix1} with an additionally condition:
\begin{equation}\label{Fuerth}
\begin{pmatrix}
     {\beta_1(f)}\\ {\beta_2(f)}
   \end{pmatrix}=\begin{bmatrix}
    \overline{a}&-\overline{b}\\
    -\overline{c}&\overline{d}
   \end{bmatrix}\begin{pmatrix}
     {\alpha_1(f)}\\
      {\alpha_2(f)}
   \end{pmatrix} \quad \mbox{ for all } f\in \dom \widetilde{A}.
\end{equation}
Assume that the matrix
\begin{equation}\label{JuniFuerth}
\begin{bmatrix}
    a&b\\
    c&d   \end{bmatrix}
\end{equation}
has a rank less than two. Then there exists $\eta \in \mathbb C$
with $c=\eta a$ and $d = \eta b$ and from \eqref{mix1}, \eqref{Fuerth}
we obtain for $f\in \dom \widetilde{A}$
$$
\alpha_2(f) =\eta \alpha_1 (f), \quad \beta_1(f) =
(\overline{a}-\overline{b}\eta)  \alpha_{1}(f)
\quad  \mbox{and} \quad \beta_2(f) =-\overline{\eta} \beta_{1}(f).
$$
But this implies that $\widetilde{A}$ is a $1$-dimensional extension
of $A$, a contradiction. Hence, the matrix in \eqref{JuniFuerth} has full rank.
Together with \eqref{mix1} it follows that for
$ f\in \dom \widetilde{A}$ we have
\begin{equation}\label{e1}
\begin{pmatrix}
     {\beta_1(f)}\\ {\beta_2(f)}
   \end{pmatrix}=\begin{bmatrix}
    \overline{a}&-\overline{b}\\
    -\overline{c}&\overline{d}
   \end{bmatrix}\begin{bmatrix}
    {a}& {b}\\
     {c}& {d}
   \end{bmatrix}\begin{pmatrix}{\beta_1(f)}\\ {\beta_2(f)}
   \end{pmatrix},
\end{equation}
There are two vectors $f,g\in\dom\widetilde{A}$ such that the vectors
\begin{equation}\label{Nuernberg}
\begin{pmatrix}
     {\beta_1(f)}\\
      {\beta_2(f)}
\end{pmatrix}\quad \mbox{and} \quad
\begin{pmatrix}{\beta_1(g)}\\ {\beta_2(g)}
   \end{pmatrix} \quad \mbox{are linearly independent.}
   \end{equation}
Indeed, assume that all such vectors are linearly dependent. Since
$\dom \widetilde{A}\neq \dom A$ and by \eqref{Sevilla} there is a vector $f_0\in\dom \widetilde{A}$ such that $|\beta_1(f_0)|+|\beta_2(f_0)|\neq 0.$ Then for each $f\in\dom\widetilde{A}$ there exists a number $\lambda(f) \in \mathbb C$ with
\begin{equation}\label{eqbeta}
  \begin{pmatrix}
     {\beta_1(f)}\\
      {\beta_2(f)}
\end{pmatrix}=\lambda(f)\,\begin{pmatrix}
     {\beta_1(f_0)}\\
      {\beta_2(f_0)}
\end{pmatrix}.
\end{equation}
and, from \eqref{mix1} and the fact that the matrix in \eqref{JuniFuerth}
has full rank, we deduce
\begin{equation}\label{eqalpha}
  \begin{pmatrix}
     {\alpha_1(f)}\\
      {\alpha_2(f)}
\end{pmatrix}=\lambda(f)\,\begin{pmatrix}
     {\alpha_1(f_0)}\\
      {\alpha_2(f_0)}
\end{pmatrix}.
\end{equation}
 Using \eqref{eqalpha} and \eqref{eqbeta} one can conclude that the
   functions $f$ from the domain of $\widetilde{A}$ satisfy the following system:
 $$
 \begin{cases}
   \alpha_1(f)\alpha_2(f_0)-\alpha_2(f)\alpha_1(f_0)=0,\\
   \beta_1(f)\beta_2(f_0)-\beta_2(f)\beta_1(f_0)=0,
 \end{cases}
 $$
 that is, the boundary conditions are separated, a contradiction.
 Hence there are two vectors $f$ and $g$ in $\dom \widetilde{A}$ such
 that \eqref{Nuernberg} holds. Then from \eqref{e1} it follows that
 $$
 \begin{bmatrix}
    \overline{a}&-\overline{b}\\
    -\overline{c}&\overline{d}
   \end{bmatrix}\,\begin{bmatrix}
    {a}& {b}\\
     {c}& {d}
   \end{bmatrix}=I.
 $$
 Therefore the matrix
 $$
 \begin{bmatrix}
    {a}& {b}\\
     {c}& {d}
   \end{bmatrix}
 $$
 is nondegenerate and we have $\Delta:=ad-bc\neq 0$ with $|\Delta|=1$. If we set $\varphi:=\frac{1}{2}\arg \Delta$ then
   $\Delta =e^{2i\varphi}$ and we obtain
  $$
\begin{bmatrix}
    {a}& {b}\\
     {c}& {d}
   \end{bmatrix}=\begin{bmatrix}
    \overline{a}&-\overline{b}\\
    -\overline{c}&\overline{d}
   \end{bmatrix}^{-1} =\begin{bmatrix}
    \overline{d}/\overline{\Delta}&\overline{b}/\overline{\Delta}\\
    \overline{c}/\overline{\Delta}&\overline{a}/\overline{\Delta}
   \end{bmatrix}.
$$
This implies that $e^{-i\varphi}a=\overline{e^{-i\varphi}d}$, and the numbers $e^{-i\varphi}b$ and $e^{-i\varphi}c$  are real.
 Set
 $$
 \alpha := e^{-i\varphi}a, \quad \beta:= e^{-i\varphi}b\quad
 \mbox{and} \quad \gamma := e^{-i\varphi}c.
 $$
 Then we have $|\alpha|^2 -\beta\gamma=1$ and  one can rewrite boundary conditions \eqref{mix1} for the $\PT$ symmetric operator $\widetilde{A}$ in the form \eqref{mixpt}.

For the converse statement assume that
for some $\alpha \in \mathbb C$, $\beta, \gamma \in \mathbb R$ with
 $|\alpha |^2-
\beta \gamma =1$ \eqref{mixpt} is satisfied. Then, according to
  \eqref{pt}, we obtain for  $f\in \dom \widetilde{A}$
  \begin{eqnarray*}
 \begin{bmatrix}
    \alpha&\beta\\
    \gamma&\overline{\alpha}
   \end{bmatrix}\begin{pmatrix}{\beta_1(\PT f)}\\ {\beta_2(\PT f)}
   \end{pmatrix}&=&
  \begin{bmatrix}
    \alpha&\beta\\
    \gamma&\overline{\alpha}
   \end{bmatrix}
   \begin{bmatrix}
    1&0\\
    0&-1
   \end{bmatrix}
   \begin{pmatrix}{ \overline{\alpha_1(f)}}\\ { \overline{\alpha_2(f)}}
   \end{pmatrix}
   \\[1ex]
  &=&
  \begin{bmatrix}
    \alpha&\beta\\
    \gamma&\overline{\alpha}
   \end{bmatrix}
  \begin{bmatrix}
    1&0\\
    0&-1
   \end{bmatrix}
   e^{-i\varphi}
   \begin{bmatrix}
    \overline{\alpha}&\beta\\
    \gamma&\alpha
   \end{bmatrix}
   \begin{pmatrix}{ \overline{\beta_1(f)}}\\ { \overline{\beta_2(f)}}
   \end{pmatrix}\\[1ex]
   &=&e^{-i\varphi}
   \begin{bmatrix}
    1 &0 \\
    0&-1
   \end{bmatrix}
   \begin{pmatrix}{ \overline{\beta_1(f)}}\\ { \overline{\beta_2(f)}}
   \end{pmatrix} = e^{-i\varphi}
   \begin{pmatrix}{\alpha_1(\PT f)}\\ {\alpha_2(\PT f)}
   \end{pmatrix}
  \end{eqnarray*}
and, hence, $\PT \dom \widetilde{A}\subset \dom \widetilde{A}$. Then,
by $\PT^2=I$, \eqref{ptdom} holds.

\ \\
The case of a domain given by \eqref{mix2} can be proved in a similar way.
In the reasoning one just has to changes the roles of  $\alpha_1$, $\alpha_2$
and $\beta_1$, $\beta_2$.
\end{proof}

In the following corollaries we will describe the situations when two out
of the three properties  $\PT$ symmetry, selfadjointness and
$\Pt$-selfadjointness are satisfied. Due to the fact that for an extension
$\widetilde{A}$ with  domain  \eqref{separate} $\PT$ symmetry is equivalent
to $\Pt$-selfadjointness by Proposition \ref{1412}, there is only one case to consider for separated
boundary conditions.
\begin{corollary}\label{SpandauNeu}
Let $\widetilde{A}$ has the domain  \eqref{separate}.
Then $\widetilde{A}$ is $\PT$ symmetric, selfadjoint and  $\Pt$-selfadjoint if and only if
$$
\xi=\eta=1 \quad \mbox{and}\quad \alpha+\beta=0 \ (\, \hspace*{-9pt}\mod{\pi}).
$$
 \end{corollary}
In the case of mixed boundary conditions, there are more cases.
\begin{corollary}\label{Spand}
Let $\widetilde{A}$ has the domain given by \eqref{mix1}
or \eqref{mix2}. Then
\begin{itemize}
\item[\rm (1)] $\widetilde{A}$ is $\PT$ symmetric, selfadjoint and  $\Pt$-selfadjoint if and only if
\begin{equation}\label{Spandau}
a,b,c,d \in \mathbb R, \quad a=d \quad \mbox{and}\quad a^2-bc=1.
\end{equation}
\item[\rm (2)] $\widetilde{A}$ is selfadjoint and  $\Pt$-selfadjoint if and only if \eqref{Spandau} holds. In this case $\widetilde{A}$ is also
    $\PT$ symmetric.
\item[\rm (3)] $\widetilde{A}$ is selfadjoint and  $\PT$ symmetric if and only if
for some $\varphi \in \mathbb R$
    $$
     e^{-i\varphi}a, e^{-i\varphi}b, e^{-i\varphi}c, e^{-i\varphi}d \in \mathbb R, 
\quad d= e^{2i\varphi}\overline{a} \quad \mbox{and}\quad ad-bc= e^{2i\varphi}.
    $$
\item[\rm (4)] $\widetilde{A}$ is  $\Pt$-selfadjoint and $\PT$ symmetric if and only if
\begin{equation*}
b,c \in \mathbb R, \quad d=\overline{a} \quad \mbox{and}\quad |a|^2-bc=1.
\end{equation*}
\end{itemize}
 \end{corollary}
 \begin{remark}
 Corollary \ref{SpandauNeu} and Corollary \ref{Spand}, item (2), are already contained 
in \cite[Theorem 4 and Theorem 5]{AT10}. Here we use the opportunity to point
 out that the statement in \cite[Theorem 5]{AT10} is slightly incorrect.
Obviously, \eqref{Spandau} implies that the corresponding extension  $\widetilde{A}$
is $\PT$ symmetric and  selfadjoint (and at the same time $\Pt$-selfadjoint),
but the converse is not true: There are $\PT$ symmetric and  selfadjoint
 extensions  $\widetilde{A}$ which do not satisfy  \eqref{Spandau}, cf.\ 
Corollary \ref{Spand}, item (3). Hence, the correct version of \cite[Theorem 5]{AT10}
 is  Corollary \ref{Spand}, item (2).
 \end{remark}

\section{$3$-dimensional extensions} \label{Three}

Let $\widetilde{A}$ be a $3$-dimensional extension of $A$. Then there are numbers $a,b,c,d$, $|a|+|b|+|c|+|d|\ne 0$  such that
\begin{equation}\label{eq2}
\dom\widetilde{A}=\left\{
f\in\dom\cD_{\mmm}\mid a\,\alpha_1(f)+b\,\alpha_2(f)=c\,\beta_1(f)+d\,\beta_2(f)
\right\} .
\end{equation}

\begin{proposition}\label{p1}
Let $\widetilde{A}$ be a $3$-dimensional extension of $A$ with domain \eqref{eq2}. Then the following statements hold.
\begin{enumerate}
\item[$(1)$] Let $a\neq 0$. Then
\begin{equation}\label{3a}
 \dom\widetilde{A}^*=\left\{
 g\in\cD_{\mmm}\left|
 \begin{array}{rcl}
  \overline{a}\, \alpha_1(g)+\overline{b}\, \alpha_2(g)&=&0,\\
  &&\\
 \begin{bmatrix}
    0&-\overline{d}/ \overline{a}\\
    0&\overline{c}/ \overline{a}
  \end{bmatrix}\,\begin{pmatrix}
    \alpha_1(g)\\
    \alpha_2(g)
  \end{pmatrix}&=& \begin{pmatrix}
    \beta_1(g)\\
    \beta_2(g)
  \end{pmatrix}
 \end{array}
 \right.
 \right\}.
\end{equation}

\begin{equation}\label{3a+}
 \dom\widetilde{A}^+=\left\{
 h\in\cD_{\mmm} \left|
 \begin{array}{rcl}
  \overline{a}\, \beta_1(h)-\overline{b}\, \beta_2(h)&=&0,\\
  &&\\
 \begin{bmatrix}
    0&\overline{d}/ \overline{a}\\
    0&\overline{c}/ \overline{a}
  \end{bmatrix}\,\begin{pmatrix}
    \beta_1(h)\\
    \beta_2(h)
  \end{pmatrix}&=& \begin{pmatrix}
    \alpha_1(h)\\
    \alpha_2(h)
  \end{pmatrix}
 \end{array}
 \right.
 \right\}.
\end{equation}

\item[$(2)$] Let $b\neq 0$.
Then
\begin{equation}\label{3b}
 \dom\widetilde{A}^*=\left\{
 g\in\cD_{\mmm} \left|
 \begin{array}{rcl}
  \overline{a}\, \alpha_1(g)+\overline{b}\, \alpha_2(g)&=&0,\\
  &&\\
 \begin{bmatrix}
    \overline{d}/ \overline{b}&0\\
    -\overline{c}/ \overline{b}&0
  \end{bmatrix}\,\begin{pmatrix}
    \alpha_1(g)\\
    \alpha_2(g)
  \end{pmatrix}&=& \begin{pmatrix}
    \beta_1(g)\\
    \beta_2(g)
  \end{pmatrix}
 \end{array}
 \right.
 \right\}.
\end{equation}

\begin{equation}\label{3b+}
 \dom\widetilde{A}^+=\left\{
 h\in\cD_{\mmm} \left|
 \begin{array}{rcl}
  \overline{a}\, \beta_1(h)-\overline{b}\, \beta_2(h)&=&0,\\
  &&\\
 \begin{bmatrix}
    \overline{d}/ \overline{b}&0\\
    \overline{c}/ \overline{b}&0
  \end{bmatrix}\,\begin{pmatrix}
    \beta_1(g)\\
    \beta_2(g)
  \end{pmatrix}&=& \begin{pmatrix}
    \alpha_1(g)\\
    \alpha_2(g)
  \end{pmatrix}
 \end{array}
 \right.
 \right\}.
\end{equation}

\item[$(3)$]\ Let $c\neq 0$. Then
\begin{equation}\label{3c}
 \dom\widetilde{A}^*=\left\{
 g\in\cD_{\mmm} \left|
 \begin{array}{rcl}
  \overline{c}\, \beta_1(g)+\overline{d}\, \beta_2(g)&=&0,\\
  &&\\
 \begin{bmatrix}
    0&-\overline{b}/ \overline{c}\\
    0&\overline{a}/ \overline{c}
  \end{bmatrix}\, \begin{pmatrix}
    \beta_1(g)\\
    \beta_2(g)
  \end{pmatrix}&=&\begin{pmatrix}
    \alpha_1(g)\\
    \alpha_2(g)
  \end{pmatrix}
 \end{array}
 \right.
 \right\}.
\end{equation}

\begin{equation}\label{3c+}
 \dom\widetilde{A}^+=\left\{
 h\in\cD_{\mmm}\left|
 \begin{array}{rcl}
  \overline{c}\, \alpha_1(h)-\overline{d}\, \alpha_2(h)&=&0,\\
  &&\\
 \begin{bmatrix}
    0&\overline{b}/ \overline{c}\\
    0&\overline{a}/ \overline{c}
  \end{bmatrix}\, \begin{pmatrix}
    \alpha_1(h)\\
    \alpha_2(h)
  \end{pmatrix}&=&\begin{pmatrix}
    \beta_1(h)\\
    \beta_2(h)
  \end{pmatrix}
 \end{array}
 \right.
 \right\}.
\end{equation}

\item[$(4)$]\ Let $d\neq 0$. Then
\begin{equation}\label{3d}
 \dom\widetilde{A}^*=\left\{
 g\in\cD_{\mmm} \left|
 \begin{array}{rcl}
  \overline{c}\, \beta_1(g)+\overline{d}\, \beta_2(g)&=&0,\\
  &&\\
 \begin{bmatrix}
    \overline{b}/ \overline{d}&0\\
    -\overline{a}/ \overline{d}&0
  \end{bmatrix}\, \begin{pmatrix}
    \beta_1(g)\\
    \beta_2(g)
  \end{pmatrix}&=&\begin{pmatrix}
    \alpha_1(g)\\
    \alpha_2(g)
  \end{pmatrix}
 \end{array}
 \right.
 \right\}.
\end{equation}

\begin{equation}\label{3d+}
 \dom\widetilde{A}^+=\left\{
 h\in\cD_{\mmm} \left|
 \begin{array}{rcl}
  \overline{c}\, \alpha_1(h)-\overline{h}\, \alpha_2(h)&=&0,\\
  &&\\
 \begin{bmatrix}
    \overline{b}/ \overline{d}&0\\
    -\overline{a}/ \overline{d}&0
  \end{bmatrix}\, \begin{pmatrix}
    \alpha_1(h)\\
    \alpha_2(h)
  \end{pmatrix}&=&\begin{pmatrix}
    \beta_1(h)\\
    \beta_2(h)
  \end{pmatrix}
 \end{array}
 \right.
 \right\}.
\end{equation}

\end{enumerate}
\end{proposition}

\begin{proof}
Let us prove $(1)$. The others one can be shown in a similar manner.

Since $a\neq 0$ we can express $\alpha_1(f)$ for $f \in \dom\widetilde{A}$:
$$
\alpha_1(f)=-\frac{b}{a}\ \alpha_2(f)+\frac{c}{a}\ \beta_1(f)+
\frac{d}{a}\ \beta_2(f).
$$
Then, by \eqref{w}, $g\in \dom  \widetilde{A}^*$ if and only if
$$
\beta_2(f)\left(\overline{\beta_1(g)}+\frac{d}{a}\ \overline{\alpha_2(g)} \right) - \beta_1(f)\left(\overline{\beta_2(g)}-\frac{c}{a}\ \overline{\alpha_2(g)} \right)=\alpha_2(f)\left(\overline{\alpha_1(g)}+\frac{b}{a}\ \overline{\alpha_2(g)} \right)
$$
for all $f\in \dom  \widetilde{A}$.
Then by Lemma \ref{HornoEnSeville}, there exists $f\in \dom \widetilde{A}$
such that
\begin{eqnarray*}
\begin{split} \alpha_1(f)&= -\frac{b}{a},\\
\beta_1(f)&= 0,
\end{split}
\quad & \quad
\begin{split}
\alpha_2(f)&= 1 , \\
 \beta_2(f)&=0
\end{split}
\end{eqnarray*}
and, hence, $g\in \dom  \widetilde{A}^*$  has to satisfy
$$
\overline{\alpha_1(g)}+\frac{b}{a}\ \overline{\alpha_2(g)} =0.
$$
In a similar way, we obtain
$$
\overline{\beta_1(g)}+\frac{d}{a}\ \overline{\alpha_2(g)} =0
\quad \mbox{and} \quad
\overline{\beta_2(g)}-\frac{c}{a}\ \overline{\alpha_2(g)} =0
$$
and  \eqref{3a} is proved. For a proof of \eqref{3a+} we use the relation $\dom\widetilde{A}^+=\Pt\dom\widetilde{A}^*$ and Lemma \ref{alpha_1}.
\end{proof}

\begin{proposition}\label{pr}
The $3$-dimensional extension $\widetilde{A}$ with domain \eqref{eq2}  is a $\PT$ symmetry if and only if
\begin{equation}\label{eq3}
|a|=|c|,\quad |b|=|d|,\quad {\mbox and }\quad  a\overline{d}+b\overline{c}=0.
\end{equation}
\end{proposition}
\begin{proof}
We assume that the relation \eqref{eq3} holds and show that  $\widetilde{A}$ is a $\PT$-symmetry, or, what is equivalent (see Lemma \ref{Jetztdoch}), $\PT \dom  \widetilde{A}=\dom
 \widetilde{A}$. Since $(\PT)^2=I$ it is sufficient to show that
 $\PT \dom  \widetilde{A}\subset \dom
 \widetilde{A}$, that is, \eqref{eq2} implies for $f\in \dom  \widetilde{A}$
 $$
 a\,\alpha_1(\PT f)+b\,\alpha_2(\PT f)=c\,\beta_1(\PT f)+d\,\beta_2(\PT f),
 $$
 or, equivalently (see Lemma \ref{alpha_1})

 \begin{equation}\label{eq4}
 \overline{c}\,\alpha_1(f)-\overline{d}\,\alpha_2(f)=\overline{a}\,\beta_1(f)
 -\overline{b}\,\beta_2(f).
 \end{equation}
Consider 3 cases:
\begin{enumerate}
 \item[(i)]\ $a=c=0$.  Then from $|b|=|d|$ it follows directly that \eqref{eq3} implies \eqref{eq4}.
 \item[(ii)]\ $b=d=0$. Analogously to $(i)$, from $|a|=|c|$ it follows directly  that \eqref{eq3} implies \eqref{eq4}.
 \item[(iii)]\ $abcd\not= 0$. In this case one can rewrite \eqref{eq3} and \eqref{eq4} in forms \eqref{eq5} and \eqref{eq6} respectively:
  \begin{eqnarray}
   -\frac{a}{b}\,(\alpha_1(f)-\frac{c}{a}\,\beta_1(f))&=&\alpha_2(f)-\frac{d}{b}\,
   \beta_2(f),\label{eq5}\\
   &&\nonumber \\
   \frac{\overline{c}}{\overline{d}}\,(\alpha_1(f)-
   \frac{\overline{a}}{\overline{c}}\,
   \beta_1(f))&=&\alpha_2(f)-\frac{\overline{b}}{\overline{d}}\,
   \beta_2(f).\label{eq6}
  \end{eqnarray}
\end{enumerate}
Now from \eqref{eq3} it follows that \eqref{eq5} implies \eqref{eq6}. Therefore $\widetilde{A}$ is a $\PT$ symmetry.\\

To prove the  converse assume that both \eqref{eq2} and \eqref{eq4} hold.
 Then (cf.\ Lemma \ref{Jetztdoch}) there exists a function $f_1\in\dom \widetilde{A}$ with $\alpha_1(f_1)=\beta_1(f_1)=0$ and
 $|\alpha_2(f_1)|+|\beta_2(f_1)|\ne 0$ and we have
\begin{eqnarray*}
b\,\alpha_2(f_1)- {d}\,
   \beta_2(f_1)&=&0,\\
 &&\\
\overline{d}\,\alpha_2(f_1)- {\overline{b}} \,
   \beta_2(f_1)&=&0.
\end{eqnarray*}
Since at least one of the numbers $\alpha_2(f_1),\ \beta_2(f_1) $ is nonzero we have   $|b|=|d|$.
Analogously using a function $f_2\in\dom \widetilde{A}$ such that $\alpha_2(f_2)=\beta_2(f_2)=0$  and $|\alpha_1(f_2)|+|\beta_1(f_2)|\ne 0$ we obtain $|a|=|c|$.\\

If $abcd=0$ the equality $a\overline{d}+b\overline{c}=0$ is trivial.\\

Let $abcd\not =0$. Consider a vector $f_3\in\dom \widetilde{A}$  such that
${\alpha_1(f_3)-
   \frac{c}{a}\,
   \beta_1(f_3)\not =0}$ and
${\alpha_2(f_3)-\frac{d}{b}\,
   \beta_2(f_3)\not =0}$. Then from \eqref{eq5} and \eqref{eq6} it follows that ${a\overline{d}+b\overline{c}=0}$.
\end{proof}

The operator  $\widetilde{A}$ is a $3$-dimensional extension of $A$ but
the kernel of $A^*-\lambda$ for non-real $\lambda$ equals $2$, see Lemma \ref{ArnstadtSued}, and we obtain
$\mathbb C \setminus \mathbb R \subset \sigma_p(\widetilde{A})$. Hence
the resolvent set of $\widetilde{A}$ is empty and the following
theorem is shown.
\begin{theorem}\label{pr2}
Let $\widetilde{A}$ be a $3$-dimensional extension of $A$ with domain \eqref{eq2}. Then
$$
\sigma(\widetilde{A}) = \mathbb C.
$$
In particular (cf.\ Proposition \ref{pr}), there are $\PT$ symmetric
 $3$-dimensional extensions of $A$ with empty resolvent set.
\end{theorem}

\section{$1$-dimensional extensions}\label{Four}

The domain of a $1$-dimensional extension $\widetilde{A}$ is defined by $3$ independent relations between $\alpha_1(f)$, $\alpha_2(f)$, $\beta_1(f)$, $\alpha_1(f)$ for $f\in \dom\widetilde{A}$. From Proposition \ref{ext} it follows that $\widetilde{A}$ is $1$-dimensional extension of $A$ if and only if its adjoint is $3$-dimensional extension of $A$; it remains to apply  Proposition \ref{p1}. Hence we have two different cases.
  \begin{equation}\label{1a}
 (I)\quad \dom\widetilde{A}=\left\{
 f\in\cD_{\mmm} \left|
 \begin{array}{rcl}
   {a_1}\, \alpha_1(f)+ {b_1}\, \alpha_2(f)&=&0,\quad |a_1|+|b_1|\neq 0,\\
  &&\\
 \begin{bmatrix}
    \alpha &\beta \\
    \gamma&\delta
  \end{bmatrix}\,\begin{pmatrix}
    \alpha_1(f)\\
    \alpha_2(f)
  \end{pmatrix}&=& \begin{pmatrix}
    \beta_1(f)\\
    \beta_2(f)
  \end{pmatrix}
 \end{array}
 \right.
 \right\}.
\end{equation}

\begin{equation}\label{2a}
(II)\quad \dom\widetilde{A}=\left\{
 f\in\cD_{\mmm} \left|
 \begin{array}{rcl}
  c_1\, \beta_1(g)+d_1\, \beta_2(g)&=&0,\quad |c_1|+|d_1|\neq 0,\\
  &&\\
 \begin{bmatrix}
     \alpha &\beta \\
    \gamma&\delta
  \end{bmatrix}\,\begin{pmatrix}
    \beta_1(f)\\
    \beta_2(f)
  \end{pmatrix}&=& \begin{pmatrix}
    \alpha_1(g)\\
    \alpha_2(g)
  \end{pmatrix}
 \end{array}
 \right.
 \right\}.
\end{equation}

One can check directly using \eqref{w} that the following proposition is true.

\begin{proposition}\label{1*}
\begin{enumerate}
\item[$(i)$]\ Let $\widetilde{A}$ has domain \eqref{1a}. Then
\begin{equation*}
  \dom\widetilde{A}^*=\left\{
  g\in\cD_{\mmm}\mid a\,\alpha_1(g)+b\,\alpha_2(g)+c\,\beta_1(g)+d\,\beta_2(g)=0,
  \right\}
  \end{equation*}
  with
   \begin{equation}\label{lll2l}
  a=\overline{a_1},\quad b=\overline{b_1},\quad \mbox{ and }\quad
  \begin{pmatrix}
    c\\ d
  \end{pmatrix}=\begin{bmatrix}
    \overline{\delta}&-\overline{\gamma}\\
    -\overline{\beta}&\overline{\alpha}
  \end{bmatrix}\,\begin{pmatrix}
    \overline{a_1}\\ \overline{b_1}
  \end{pmatrix}.
  \end{equation}
\begin{equation*}
  \dom\widetilde{A}^+=\Pt\dom\widetilde{A}^*= \left\{
  h\in\cD_{\mmm}\mid a\,\beta_1(h)-b\,\beta_2(h)+c\,\alpha_1(h)-d\,\alpha_2(h)=0,
  \right\}
  \end{equation*}
  where $   a,\ b,\ c,\ d$  are the same as above.

\item[$(ii)$]\ Let $\widetilde{A}$ has domain \eqref{2a}. Then
\begin{equation*}
  \dom\widetilde{A}^*=\left\{
  g\in\cD_{\mmm}\mid a\,\alpha_1(g)+b\,\alpha_2(g)+c\,\beta_1(g)+d\,\beta_2(g)=0,
  \right\}
  \end{equation*}
  with
  \begin{equation}\label{lll1l}
  c=\overline{c_1},\quad d=\overline{d_1},\quad \mbox{ and }\quad
  \begin{pmatrix}
    a\\ b
  \end{pmatrix}=\begin{bmatrix}
    \overline{\delta}&-\overline{\gamma}\\
    -\overline{\beta}&\overline{\alpha}
  \end{bmatrix}\,\begin{pmatrix}
    \overline{c_1}\\ \overline{d_1}
  \end{pmatrix}.
  \end{equation}
  \begin{equation*}
  \dom\widetilde{A}^+=\Pt\dom\widetilde{A}^*=\left\{
  h\in\cD_{\mmm}\mid a\,\beta_1(h)-b\,\beta_2(g)+c\,\alpha_1(g)-d\,\alpha_2(g)=0,
  \right\}
  \end{equation*}
where $   a,\ b,\ c,\ d$  the same as above.
\qed
\end{enumerate}

\end{proposition}

\begin{remark}
From Lemma \ref{la} and Proposition \ref{pr} follow that a $1$-dimensional extension $\widetilde{A}$ is $\PT$-symmetry if and only if the numbers
$a, b, c ,d$
defined in Proposition \ref{1*}
 \eqref{lll2l} and \eqref{lll1l}, respectively, satisfy \eqref{eq3}.

\end{remark}

If $\widetilde{A}$ is a $1$-dimensional extension of $A$, then
$\widetilde{A}^*$ is a $3$-dimensional extension of $A$ with empty
resolvent set, see Theorem \ref{pr2} and we obtain the following.
\begin{theorem}
Let $\widetilde{A}$ be a $1$-dimensional extension of $A$. Then
$$
\sigma(\widetilde{A}) = \mathbb C.
$$
In particular, there are $\PT$ symmetric
 $1$-dimensional extensions of $A$ with empty resolvent set.
\end{theorem}

\subsection*{Acknowledgement}
The authors thank Johannes Sj\"{o}strand for fruitful discussions which
were the starting point for the present paper.

Tomas Ya.\ Azizov\\
Department of Mathematics\\
 Voronezh State University\\
  Universitetskaya pl.~1\\
   394006 Voronezh\\
 Russia\\
 azizov@math.vsu.ru

\ \\
Carsten Trunk\\
Institut f\"{u}r Mathematik\\
 Technische Universit\"{a}t
 Ilmenau\\
   Postfach 10 05 65\\
    D-98684~Ilmenau\\
    Germany\\
carsten.trunk@.tu-ilmenau.de

\end{document}